\documentclass[%
reprint,
amsmath,amssymb,
aps,
]{revtex4-2}

\usepackage{graphicx}
\usepackage{dcolumn}
\usepackage{amsmath,amsfonts,amssymb,amsthm,epsfig,amscd,comment,latexsym,psfrag}
\usepackage{nicefrac,xspace,tikz}
\usepackage{physics}
\usepackage{booktabs}
\usepackage{makecell}
\usepackage{amsfonts}
\usepackage{footmisc}
\usepackage[colorlinks=true,linkcolor=blue,citecolor=blue,urlcolor=blue]{hyperref}
\usetikzlibrary{arrows.meta, positioning,shapes.geometric}
\usepackage{bm}
\usepackage{pgfplots}
\pgfplotsset{compat=1.17}
\usepackage{siunitx}
\usepackage {setspace}

\RequirePackage{fix-cm}

\def\D{\mathsf D}

\def\C{\mathbb C}

\def\1{{\bf{1}}}

\def \<{\langle}
\def \>{\rangle}

\usepackage{etex}
\usepackage{graphicx}
\usepackage{amsmath}
\usepackage{floatflt}
\usepackage{amsfonts}
\usepackage{amssymb}
\usepackage{graphicx}
\usepackage{psfrag}
\usepackage{epsf}
\usepackage{slashed}
\usepackage{booktabs}
\usepackage{mathptmx}
\usepackage{latexsym}
\usepackage{subfigure}
\usepackage [latin1]{inputenc}
\def\footnoterule{\kern 1mm \hrule width 7cm \kern 2.2mm}%

\newcommand{\bea}{\begin{eqnarray}}
\newcommand{\eea}{\end{eqnarray}}
\newcommand{\be}{\begin{equation}}
\newcommand{\ee}{\end{equation}}
\newcommand{\on}{\operatorname}
\usetikzlibrary{arrows}
\usetikzlibrary{decorations.markings}

\newcommand{\kb}[2]{\ketbra{#1}{#2}}
\newcommand{\bk}[2]{\braket{#1}{#2}}

\def\D{\mathsf D}

\def\C{\mathbb C}

\def\1{{\bf{1}}}

\def \<{\langle}
\def \>{\rangle}
\def\footnoterule{\kern 1mm \hrule width 7cm \kern 2.2mm}

\newtheorem{thm}{Theorem}[section]
\newtheorem{prop}[thm]{Proposition}

\newtheorem{eg}[thm]{Example}
\newtheorem{rem}[thm]{Remark}
\renewcommand{\footnoterule}{\vspace*{-1mm}\hrule width 3cm height 0.1mm\vspace*{5mm}}

\begin{document}
	
\preprint{APS/123-QED}
	
\title{Quantum-state block texture and its quantification}
\author{ Yanjun Chu$^{1,*}$, Yuhang Xie$^{1}$, Chenyang Cui$^{1}$, Shao-Ming Fei$^{2,\dagger}$\\
1. School of Mathematics and Statistics, Henan University, \\
Kaifeng, 475004, China\\
2. School of Mathematics Sciences, Capital Normal University, \\Beijing 100048, China}
\date{\today}
\thanks{Contact author: chuyj@henu.edu.cn;
$^\dagger$Contact author:feishm@cnu.edu.cn}	

\begin{abstract}
Quantum-state texture (QST) is an emerging quantum resource that has garnered increasing attention amid advances in quantum theory. In this work, we generalize the QST to quantum-state block texture (QSBT). This generalization provides profound operational interpretations for quantifying the advantages of quantum states in quantum information processing. We pioneer an alternative framework for characterizing and quantifying quantum-state block texture, and propose three types of block texture measures. By comparing these QSBT measures, we investigate their distinctions and interrelationships. We demonstrate that the geometric measure serves as an upper bound for the trace distance-based measure. For a specific family of quantum states, we evaluate the values of two trace distance-based measures. Then we sample four sets of data from this family of states, with each set comprising $5\times 10^4, 10^5, 5\times 10^5$ and $10^6 $ samples, respectively, and present the corresponding distributions. Our results reveal that the QSBT measures constructed via different approaches show distinct characteristics, indicating their potential roles in quantifying the block texture of quantum states.
\end{abstract}
	
\maketitle
	
	
\section{INTRODUCTION}
The notion of quantum state texture (QST) was first proposed by Parisio \cite{parisio 1} as a novel quantum resource intrinsically associated with coherence. A straightforward interpretation of QST is as follows: For a given computational basis $\{|i\rangle\}$, the density matrix $\rho$ can be visualized as a three-dimensional plot, where the row and column indices of the matrix correspond to the horizontal dimensions, and the magnitudes of the real and imaginary parts of each matrix entry $\rho_{jk}$ serve as the vertical  altitude in two complementary sub-plots. In this geometric framework, each quantum state is mapped to a pair of real-valued three-dimensional plots (one for the real part and one for the imaginary part  of $\rho$), and the unevenness of these plots--deviations from a flat distribution of matrix entries--characterizes the QST of the state.
	
From this geometric perspective, the resulting plots typically exhibit unevenness--an attribute that defines the texture of the quantum state. Among all such plots, the simplest corresponds to a perfectly flat distribution, where the altitude (magnitude of matrix entries) is uniform across all coordinates. The only existing quantum state which gives rise to such a plot is
\begin{equation}\label{e1}
f_1 = |f_1\rangle\langle f_1|,
\end{equation}
where $|f_1\rangle = \frac{1}{\sqrt{d}} \sum\limits_{i=0}^{d-1} |i\rangle$ and $d$ is the dimension of Hilbert space. We refer to this state as textureless, and it uniquely constitutes the textureless state set  of the proposed theory. Parisio \cite{parisio 1} demonstrated that circuit layers containing at least one CNOT gate can be fully characterized using randomized input states and analyzing the texture of the output qubits. Notably, this characterization process obviates the need for tomographic protocols or ancillary systems. Therefore, QST shares similarities with quantum coherence, which is a key component in emerging quantum technologies, such as quantum metrology \cite{giovannetti 2,giovannetti 3}, quantum computing \cite{hillery 4}, nanoscale thermodynamics \cite{gour 5, lostaglio 6, narasimhachar 7, francica 8} and biological systems \cite{parisio 1,lloyd 9,giovannetti 11}. 

Recently, the quantification of QST has attracted significant interest as a novel quantum resource theory. Several potential QST measures have been proposed in Ref. \cite{wang2025}. The authors in Ref. \cite{zhang2025} have explored the quantification of QST and established uncertainty relations for multiple candidate measures based on the Tsallis relative entropy and weight. R. Muthuganesan \cite{Muthuganesan2025} developed and analyzed several new QST measures based on the Hellinger distance, Jensen-Shannon divergence and Wigner-Yanase-Dyson skew information, and established a no-broadcasting theorem for the Jensen-Shannon divergence-based QST measure, highlighting its operational significance in information-processing tasks. Meanwhile, the authors of Ref. \cite{cui2025} addressed three key aspects of the resource theory of quantum-state texture: the quantification of QST, the transformation of quantum states under free operations, and the interrelationships among these quantum resources.
 
Generally, a quantum resource theory contains three basic elements: free states, free operations and resource measures (monotones). In the resource theory of coherence, the free states are incoherent states, i.e., diagonal states with respect to a fixed orthonormal basis \cite{baumgratz2014}. In the resource theory of imaginarity, the free states are real states under a fixed orthonormal basis \cite{Hickey2018}. Endowed with a rich set of free states, these resource theories give rise to diverse free operations and resource measures, thereby facilitating their broader applications in quantum information processing. As a novel resource theory, QST  also encompasses these three fundamental components \cite{parisio 1}. However, the only textureless state is $f_1$ (\ref{e1}) in quantum texture resource theory. Its free operations are also constrained to preserve this unique textureless state, rendering the quantum texture resource theory quite different from the resource theories of coherence and  imaginarity.

In this work, analogous to the theories of POVM coherence \cite{Bischof2019} and Block coherence \cite{Xu2020}, we generalize the concept of quantum-state texture, a quantum resource describing the ``unevenness" of density matrix entries in a given orthonormal basis, to the context of quantum-state block texture (QSBT), by extending the set of free states originally consisting of a single element $f_1$ (the textureless state) to a set composed of infinitely many block matrices with block textureless states as the zero-resource set.

The QSBT is valid for arbitrary single and multipartite quantum systems.  In this broader framework, we quantify QSBT by proposing several QSBT measures. For any given QST measure, we may construct a corresponding QSBT measure in terms of the reduced states of the subsystem. Moreover, we propose quantifiers for QSBT based on the concave functions and convex roof construction, a traditional and effective approach in the quantification of entanglement \cite{vidal2000}, coherence \cite{du2015} and imaginarity \cite{du2025}. In general we also present QSBT measures including the geometric measure of QSBT ($M_g$), the trace distance-based measure of QSBT ($M_{\tr}$), the fidelity-based measure of QSBT ($M_F$) and the relative entropy of QSBT ($M_r$). We demonstrate that the geometric measure $M_g$ serves as an upper bound for the trace distance-based measure $M_{\tr}$. For a specific family of quantum states, we evaluate two trace distance-based measures introduced in this work. We sample four sets of data from this family of states, with each set comprising $5\times 10^4, 10^5, 5\times 10^5$ and $10^6 $ samples, respectively. The distributions corresponding to these four sets are analyzed accordingly.  Our results show that the QSBT measures constructed via different approaches are distinct, indicating that they are poised to fulfil distinct roles in quantifying the block texture of quantum states.

\section{Quantum-state block texture of a general  quantum system}\label{sec:basics}

We generalize the definition of  the quantum-state texture to the quantum-state block texture.
Let $\mathcal{H}$ be an $N$-dimensional Hilbert space, with $\left|x_i\right\rangle$, $i=1, \ldots, N$, being an orthonormal basis. Denote $\on{M}(\mathcal{H})$  the  algebra consisting of all matrices and $\mathcal{D}(\mathcal{H})$ the convex set consisting of all quantum states (density matrices) on the Hilbert space $\mathcal{H}$.
Given a factor $n$ of $N$, let $m=\frac{N}{n}$ of $N$ be another factor. We define a family of quantum states, 
\begin{equation}\label{e04}
f_1^{G}=\frac{1}{m}
\left(
\begin{array}{llll}
G & G & \cdots & G\\
G & G & \cdots & G\\
\ \vdots & \ \vdots & \cdots & \ \vdots\\
G & G & \cdots &G
\end{array}
\right)
\end{equation}
for arbitrary  $n\times n$ matrices $G$.

\begin{prop}\label{p2.1}
$f_1^{G}$ is a quantum state in $\mathcal{H}$ if and only if  $G$ is positive semi-definite with unit trace, i.e., 
$G$ is an $n\times n$ density matrix.
\end{prop}

\begin{proof} Since $f_1^G$ is a quantum state, we have $\Tr(G)=1$ by $\Tr(f_1^G)=1$. From
$$
\begin{array}{llll}
\left(
\begin{array}{llll}
\ \ I_n & 0 & \cdots & 0 \\
-I_n & I_n & \cdots &0 \\
\ \ \vdots & \ \vdots & \cdots & \ \vdots \\
-I_n & 0 & \cdots & I_n
\end{array}
\right)
\left(
\begin{array}{llll}
G & G & \cdots & G \\
G & G & \cdots & G \\
\ \vdots & \ \vdots & \cdots & \ \vdots \\
G & G & \cdots & G
\end{array}
\right)
\left(
\begin{array}{llll}
I_n & -I_n & \cdots &-I_n \\
0 & \ \ I_n & \cdots & \ \ 0 \\
\vdots & \ \ \vdots & \cdots & \ \ \vdots \\
0 & \ \ 0 & \cdots & \ \ I_n
\end{array}
\right)\\
=
\left(
\begin{array}{llll}
G & \ \ 0 & \cdots & \ \ 0 \\
0 & \ \ 0 & \cdots & \ \ 0 \\
\vdots & \ \ \vdots & \cdots & \ \ \vdots \\
0 & \ \ 0 & \cdots &\ \ 0
\end{array}
\right),
\end{array}
$$
we obtain that $f_1^G$ is a positive semi-definite matrix if and only if $G$ is positive semi-definite, where $I_n$ denotes the $n\times n$ identity matrix.
\end{proof}

Given an arbitrary $n\times n$ density matrix $G$, the corresponding quantum 
state $f_1^G$  is called block textureless on $ \mathcal{H}$.
We denote the set of all block textureless quantum states in $\mathcal{H}$ by 
$\mathcal{T}_B(\mathcal{H})$, 
\begin{equation}\label{e05}
\mathcal{T}_B(\mathcal{H})=\left\{
f_1^G|  G\in \mathcal{D}(\C^n)\right\},
\end{equation}
where $\mathcal{D}(\C^n)$ is the set of all density matrices on $\C^n$.
Obviously, when $n=1$, the quantum-state block texture reduces to the quantum-state texture given in 
Ref. \cite{parisio 1}. 

\section{An alternative framework for quanfying the block texture}
We can naturally introduce the concept of quantum-state block texture in bipartite systems. 
Let $\mathcal{H}_A$ (respectively $\mathcal{H}_B$ ) be an $m$ ($n$ )-dimensional Hilbert space, with $\left|e_i\right\rangle$, $i=1, \ldots, m$ ($\left|h_j\right\rangle$, $j=1, \ldots, n$ ) as an orthonormal basis.
Denote
\begin{equation}\label{e3}
f_1^A=\frac{1}{m}\sum\limits_{i=1}^{m}\sum\limits_{j=1}^{m}|e_i\rangle\langle e_j|.
\end{equation}
The bipartite state $f_1^A\otimes \rho_B$ in $\mathcal{H}_A \otimes \mathcal{H}_B$ has the following  form 
\begin{equation}\label{e4}
f_1^A\otimes \rho_B=\frac{1}{m}
\left(
\begin{array}{llll}
\rho_B & \rho_B & \cdots & \rho_B\\
\rho_B & \rho_B & \cdots & \rho_B\\
\ \ \vdots & \ \ \vdots & \cdots & \ \ \vdots\\
\rho_B & \rho_B & \cdots & \rho_B
\end{array}
\right),
\end{equation}
which is  block textureless in $\mathcal{H}_A \otimes \mathcal{H}_B$ with respect to $\left|e_i\right\rangle$, $i=1, \ldots, d_A$ of the subsystem $\mathcal{H}_A$. We define the set of all block textureless quantum states on $\mathcal{H}_A \otimes \mathcal{H}_B$ by $\mathcal{T}_B(\mathcal{H}_A\otimes\mathcal{H}_B)$, 
\begin{equation}\label{e5}
\mathcal{T}_B(\mathcal{H}_A\otimes\mathcal{H}_B)=\left\{
f_1^A\otimes \rho_B|\rho_B\in \mathcal{D}(\mathcal{H}_B)\right\}.
\end{equation}
In the framework of quantum-state block texture resource theory, $\mathcal{T}_B(\mathcal{H}_A\otimes\mathcal{H}_B)$ is the set of free states.  For arbitrary $\lambda\in [0,1]$, we have 
\begin{align*}
\lambda f_1^A\otimes \rho^1_B+&(1-\lambda )f_1^A\otimes \rho^2_B\\&=f_1^A\otimes (\lambda\rho^1_B+(1-\lambda )\rho^2_B\in \mathcal{T}_B(\mathcal{H}_A\otimes\mathcal{H}_B).
\end{align*}
Hence, the set $\mathcal{T}_B(\mathcal{H}_A\otimes\mathcal{H}_B)$ is a convex set.

We next establish a framework for quanfying the block texture  of bipatite quantum systems.
 For an arbitrary quantum-state block texture (QSBT) measure $\mathcal{P}$,  it must satisfy the following conditions:

(T1) Nonnegativity: $\mathcal{P}(\rho^{AB})\geqslant 0$ for arbitrary $\rho^{AB}\in \mathcal{D}(\mathcal{H}_{A}\otimes \mathcal{H}_{B}$), and $\mathcal{P}(\sigma^{AB})=0$ if and only if $\sigma^{AB}\in \mathcal{T}_B(\mathcal{H}_A\otimes\mathcal{H}_B)$;

(T2) Monotonicity: $\mathcal{P}(\rho^{AB})\geqslant \mathcal{P}(\Phi(\rho^{AB}))$ for any quantum channel $\Phi$ such that $\Phi(\sigma^{AB})\in \mathcal{T}_B(\mathcal{H}_A\otimes\mathcal{H}_B)$ for arbitrary $\sigma^{AB}\in \mathcal{T}_B(\mathcal{H}_A\otimes\mathcal{H}_B)$;

(T3) Convexity: $\mathcal{P}(\sum\limits_ip_i\rho_i^{AB})\leqslant \sum\limits_ip_i\mathcal{P}(\rho_i^{AB})$, $\sum\limits_ip_i=1$, $p_i \geqslant 0$.
 
Analogous to quantum coherence \cite{baumgratz2014,yu2016} and quantum imaginarity \cite{xue2021}, we also investigate the following conditional desirable property:

(T4) Strong monotonicity: $\mathcal{P}(\rho^{AB})\geqslant\sum\limits_i p_i \mathcal{P}(\rho^{AB}_i)$, where $p_i=\Tr(K_i\rho^{AB}K_i^{\dagger}), \rho^{AB}_i=\flatfrac{K_i\rho^{AB}K_i^{\dagger}}{p_i}$, and $\{K_i\}$ satisfies $K_i\mathcal{T}_B(\mathcal{H}_A\otimes\mathcal{H}_B)K_i^{\dagger}\subseteq\mathcal{T}_B(\mathcal{H}_A\otimes\mathcal{H}_B)$.
%

\subsection{Quantum-state block texture measures induced by  quantum-state  texture measures}
Under the framework of quantum-state block texture (QSBT), we can  establish the connection between QSBT measures for bipartite systems and general QST measures. If $\mathcal{Q}$ is a QST measure,  given an arbitrary  $\rho^{AB}\in \mathcal{D}(\mathcal{H}_A\otimes \mathcal{H}_B)$, 
we define a function 
$\mathcal{Q}^{AB}$ on the system $\mathcal{H}_A\otimes \mathcal{H}_B$ as follows,
\begin{equation}\label{e7}
\mathcal{Q}^{AB}(\rho^{AB}):=\mathcal{Q}(\Tr_B(\rho^{AB})).
\end{equation}

Before presenting the following theorem, we first introduce a fact: for a bipartite quantum state \(\sigma^{AB}\), if its partial trace over system \(B\) satisfies
\[
\operatorname{Tr}_B\bigl(\sigma^{AB}\bigr)
=
\ket{\phi}\!\bra{\phi}_A,
\]
where $\ket{\phi}\!\bra{\phi}_A$ is a pure state on system $\mathcal{H}_A$, then the quantum state must be of the form
\[
\sigma^{AB}
=
\ket{\phi}\!\bra{\phi}_A \otimes \tau^B,
\]
where \(\tau^B\) is a quantum state on system \(\mathcal{H}_B\).

In fact, we have the spectrum decomposition of the quantum state $\sigma^{AB}$ as follows
\begin{equation*}
\sigma^{AB}=\sum_n p_n\ket{\Phi_n}\!\bra{\Phi_n}_{AB},\quad p_n\geqslant 0,\quad \sum_n p_n=1.
\end{equation*}

For each $\ket{\Phi_n}_{AB}$, by using Schimidt decomposition we can have
\(
\ket{\Phi_n}_{AB}=\sum_k\sqrt{\lambda_{n,k}}\ket{u_{n,k}}_A\otimes\ket{v_{n,k}}_B,
\)
so
\[
\sigma^{AB}=\sum_n p_n\sum_{i,j}\sqrt{\lambda_{n,i}\lambda_{n,j}}\ket{u_{n,i}}\!\bra{u_{n,j}}_A\otimes\ket{v_{n,i}}\!\bra{v_{n,j}}_B.
\]

By taking partial trace with respect to system $\mathcal{H}_B$, one have
\begin{equation*}
\Tr_B(\sigma^{AB})=\sum_{n,k}p_n \lambda_{n,k}\ket{u_{n,k}}\!\bra{u_{n,k}}_A:=\sum_m q_m\ket{u_m}\!\bra{u_m}_A,
\end{equation*}
where $m:=(n,k),q_m:=p_n\lambda_{n,k}$, and $\{q_m\}$ is a probability distribution.

Thus, if the above expression equals $\ket{\phi}\!\bra{\phi}_A$, it follows from the fact that a pure state cannot be written as a nontrivial convex combination of states that all $\ket{u_m}_A$ must be coincide with $\ket{\phi}_A$.

Consequently,
\begin{align*}
\sigma^{AB}&=\sum_n p_n\sum_{i,j}\sqrt{\lambda_{n,i}\lambda_{n,j}}\ket{f^A}\!\bra{f^A}_A\otimes\ket{v_{n,i}}\!\bra{v_{n,j}}_B \\
&=f^A\otimes\left(\sum_n p_n\sum_{i,j}\sqrt{\lambda_{n,i}\lambda_{n,j}}\ket{v_{n,i}}\!\bra{v_{n,j}}_B\right) \\
&:=f^A\otimes\tau^B.
\end{align*}

Especially, if a bipartite quantum state \(\sigma^{AB}\) satisfies $\operatorname{Tr}_B\bigl(\sigma^{AB}\bigr)=\ket{f^A}\!\bra{f^A}$, then $\sigma^{AB}=f^A\otimes\tau^B$.

\begin{thm}\label{t4.1}
$\mathcal{Q}^{AB}$ is a well-defined  QSBT measure on the system $\mathcal{H}_A\otimes \mathcal{H}_B$
if $\mathcal{Q}$ is a QST measure on the system $\mathcal{H}_A$.
\end{thm}

\begin{proof}
An arbitrary quantum-state texture measure $\mathcal{T}$ satisfies the following conditions \cite{parisio 1,wang2025,zhang2025,cui2025}: i) $\mathcal{T}(\rho) \geqslant 0$ and $\mathcal{T}(f_1) = 0$;
ii) $\mathcal{T}(\rho) \geqslant \mathcal{T}\left[\Phi(\rho)\right]$ for any
completely positive and trace-preserving (CPTP) maps $\Phi(\rho) = \sum\limits_n K_n \rho K_n^\dagger$, $\sum_n K_n^\dagger K_n = \mathbb{I}$, satisfying $\Phi(f_1) = f_1$; 
iii) $\mathcal{T}$ is convex, $\mathcal{T}\left(\sum\limits_i p_i \rho_i\right) \leqslant \sum\limits_i p_i \mathcal{T}(\rho_i)$ 
for any probability distribution $\{p_i\}$ ($\sum\limits_i p_i = 1$ and $p_i \geqslant 0$).

Obviously $\mathcal{Q}^{AB}(\rho^{AB})\geqslant 0$ for any $\rho^{AB}\in \mathcal{D}(\mathcal{H}_A\otimes \mathcal{H}_B)$.  For arbitrary block textureless state $\sigma^{AB}\in \mathcal{T}_B(\mathcal{H}_A\otimes\mathcal{H}_B)$, one has
$\mathcal{Q}^{AB}(\sigma^{AB})=\mathcal{Q}(f_1^A)=0$. Conversely, if a state $\rho^{AB}$ satisfies $\mathcal{Q}^{AB}(\sigma^{AB})=0$ and the considered QST measure $\mathcal{Q}$ is faithful, then one have $\Tr_B(\rho^{AB})=f^A$, which implies $\sigma^{AB}=f^A\otimes\tau^B$, meaning that $\sigma^{AB}\in\mathcal{T}_B(\mathcal{H}_A\otimes\mathcal{H}_B)$ is a block textureless state.

Let $\Psi: \mathcal{B}(\mathcal{H}_A\otimes\mathcal{H}_B) \to \mathcal{B}(\mathcal{H}_A\otimes\mathcal{H}_B)$
 be a quantum channel (completely positive and trace-preserving, CPTP). If $\Psi(\sigma^{AB})\in\mathcal{T}_B(\mathcal{H}_A\otimes\mathcal{H}_B)$, i.e., $\Psi(\sigma^{AB})=\Psi(f_1^A\otimes \rho^B)=f_1^A\otimes \rho^{\prime, B}$, we can define a channel  $\Psi_B(\rho^B)$ acting on the system $\mathcal{H}_B$, where $\Psi_B(\sigma^{AB})=\Tr_A(\Psi(\sigma^{AB}))=\rho^{\prime, B}$.
Since $\Psi$ is completely positive and trace preserving, then $\Psi_B$ is a well-defined quantum channel on $\mathcal{H}_B$.

Let $\{K_\mu\}$ be a Kraus representation of $\Psi$, so that
\[
\Psi(X) = \sum_\mu K_\mu X K_\mu^\dagger,
\qquad
\sum_\mu K_\mu^\dagger K_\mu = \mathbb{I}_{AB}.
\]

\begin{widetext}
\begin{table}[t]
    \centering
    \caption{Some QSBT measures induced by QST measures}
    \label{tab_01}
    \begin{tabular}{ll}
        \toprule
        QST measures & QSBT measures induced by QST measures \\
        \midrule
        (1) trace-distance measure $T_{\tr}(\rho)=\frac{1}{2}\norm{\rho-f_1}_1$ \cite{wang2025} & $T_{\tr}^{AB}(\rho^{AB}):=T_{\tr}(\Tr_B(\rho^{AB}))=\frac{1}{2}\norm{\Tr_B(\rho^{AB})-f_1^A}_1$ \\
        \addlinespace[1.5em]
        (2) geometric measure $T_g(\rho)=1-\bra{f_1}\rho\ket{f_1}$ \cite{wang2025} & $T_g^{AB}(\rho^{AB}):=T_g(\Tr_B(\rho^{AB}))=1-\bra{f_1^A}\Tr_B(\rho^{AB})\ket{f^A_1}$ \\
        \addlinespace[1.5em]
        (3) fidelity measure $T_F(\rho)=1-F(\rho,f_1)$ \cite{wang2025} & $T_F^{AB}(\rho{AB}):=T_F(\Tr_B(\rho^{AB}))=1-F(\Tr_B(\rho^{AB}),f^A_1)$ \\[1em]
        (4) Bures measure $T_B(\rho)=2\left(1-\sqrt{F(\rho,f_1)}\right)$ \cite{wang2025} & $T_B^{AB}(\rho^{AB}):=T_B(\Tr_B(\rho^{AB}))=2\left(1-\sqrt{F(\Tr_B(\rho^{AB}),f^A_1)}\right)$ \\
        \addlinespace[1.5em]
        \makecell[l]{(5) weight-based measure \cite{zhang2025} \\\quad$T_w(\rho)=\min\{s\geqslant 0\mid\rho=(1-s)f_1+s\sigma,\sigma\in\mathcal{D}(\mathcal{H})\}$} & $\begin{aligned}T_w^{AB}(\rho^{AB})&:=T_w(\Tr_B(\rho^{AB}))\\&=\min\{s\geqslant 0\mid\Tr_B(\rho^{AB})=(1-s)f^A_1+s\sigma^A,\sigma^A\in\mathcal{D}(\mathcal{H}_A)\}\end{aligned}$ \\
        \addlinespace[1.5em]
        \makecell[l]{(6) Tsallis relative entropy of QST \cite{zhang2025} \\ \quad$T_T(\rho)=\dfrac{1-\bra{f_1}\rho^\mu\ket{f_1}}{1-\mu},\quad\mu\in(0,1)$} & $T_T^{AB}(\rho^{AB}):=T_T(\Tr_B(\rho^{AB}))
        =\dfrac{1-\bra{f_1}\left(\Tr_B\rho^{AB}\right)^\mu\ket{f_1}}{1-\mu},\quad\mu\in(0,1)$ \\
        \addlinespace[1.5em]
        \makecell[l]{(7) QST measure based on Hellinger distance \cite{Muthuganesan2025}\\ \quad$T_H(\rho)=D(\rho,f_1)$, where $D(\rho,\sigma)=\Tr(\rho-\sigma)^2$ is\\ \quad the Hellinger distance between states $\rho,\sigma$.} & $T_H^{AB}(\rho^{AB}):=D(\Tr_B(\rho^{AB}),f^A_1)=\Tr(\Tr_B(\rho^{AB})-f^A_1)^2$\\
        \addlinespace[1.5em]
        \makecell[l]{(8) QST Measure via skew information \cite{Muthuganesan2025}\\ \quad$\begin{aligned}T_{\text{skew}}(\rho)&=I_{\alpha}(\rho,f_1)\\&=\bra{f_1}\rho\ket{f_1}-\bra{f_1}\rho^{\alpha}\ket{f_1}\bra{f_1}\rho^{1-\alpha}\ket{f_1},\end{aligned}$\\ \quad where $I_{\alpha}(\rho,K)=-\frac{1}{2}\Tr([\rho^{\alpha},K][\rho^{1-\alpha},K])$ is the skew\\ \quad information between quantum state $\rho$ and  observable $K$,\\ \quad $\alpha\in(0,1)$.} & $\begin{aligned}T_{\text{skew}}^{AB}(\rho^{AB})&:=T_{\text{skew}}(\Tr_B(\rho^{AB}))\\&=I_{\alpha}(\Tr_B(\rho^{AB}),f^A_1)\\&=\bra{f^A_1}\Tr_B\left(\rho^{AB}\right)\ket{f^A_1}-\bra{f^A_1}\left(\Tr_B\rho^{AB}\right)^{\alpha}\ket{f^A_1}\\&\quad\times\bra{f^A_1}\left(\Tr_B\rho^{AB}\right)^{1-\alpha}\ket{f^A_1},\quad\alpha\in(0,1)\end{aligned}$\\
         \addlinespace[1.5em]
         \makecell[l]{(9) JSD-based measure $T_J(\rho)=J(\rho,f_1)$ \cite{Muthuganesan2025}, where\\ \quad$\begin{aligned}J(\rho,\sigma)&=\frac{1}{2}S\left(\rho\middle\|\frac{\rho+\sigma}{2}\right)+\frac{1}{2}S\left(\sigma\middle\|\frac{\rho+\sigma}{2}\right)\\&=S\left(\frac{\rho+\sigma}{2}\right)-\frac{1}{2}S(\rho)-\frac{1}{2}S(\sigma)\end{aligned}$\\ \quad is the Jensen-Shannon divergence between two states\\ \quad$\rho,\sigma$ and $S(\rho)=-\Tr(\rho\log\rho)$.} & $T_J^{AB}(\rho^{AB}):=T_J(\Tr_B(\rho^{AB}))=J(\Tr_B(\rho^{AB}),f^A_1)$\\
        \bottomrule
    \end{tabular}
\end{table}
\end{widetext}

Substituting the given condition yields
\[
\sum_\mu K_\mu (f_1^A\otimes\rho_B) K_\mu^\dagger
= f_1^A \otimes \Psi_B(\rho_B),
\quad \forall\rho_B\in \mathcal{D}(\mathcal{H}_B).
\]
This identity holds for all $\rho_B$ only if each Kraus operator decouples across subsystems, i.e.,
\[
K_\mu = A_\mu \otimes B_\mu
\]
for operators $A_\mu$ on $\mathcal{H}_A$ and $B_\mu$ on $\mathcal{H}_B$.
Inserting this form back into the Kraus sum defines two channels:
\[
\mathcal{C}_A(Y) := \sum_\mu A_\mu Y A_\mu^\dagger,
\qquad
\Psi_B(Z) := \sum_\mu B_\mu Z B_\mu^\dagger.
\]
By linearity and the defining property of tensor-product maps, we can get 
\(
\Psi = \mathcal{C}_A \otimes \Psi_B.
\)
Here, the condition $\Psi(f_1^A\otimes\rho_B)=f_1^A\otimes\Psi_B(\rho_B)$
implies $\mathcal{C}_A(f_1^A)=f_1^A$.
Thus, we  have
\begin{align*}
\mathcal{Q}^{AB}(\Psi(\rho^{AB}))& =\mathcal{Q}^{AB}((\mathcal{C}_A\otimes \Psi_B)(\rho^{AB})) \\
&=\mathcal{Q}(\Tr_B[(\mathcal{C}_A\otimes \Psi_B)(\rho^{AB})]) \\
&=\mathcal{Q}((\mathcal{C}_A\otimes \Tr\Psi_B)(\rho^{AB})).
\end{align*}

Let $\rho^{AB}=\sum\limits_{i=1}^{d_A}\sum\limits_{j=1}^{d_A}
 \sum\limits_{k=1}^{d_B}\sum\limits_{l=1}^{d_B}a_{ij}^{kl}|e_i\rangle\langle j|\otimes |h_k\rangle \langle h_l|$ 
be arbitrary a bipartite state. Then
\begin{align*}
&\quad (\mathcal{C}_A\otimes \Tr\Psi_B)(\rho^{AB}) \\
&= (\mathcal{C}_A\otimes\Tr\Psi_B)\left(\sum_{i=1}^{d_A}\sum_{j=1}^{d_A}\sum_{k=1}^{d_B}\sum_{l=1}^{d_B}a_{ij}^{kl}|e_i \rangle\langle e_j|\otimes |h_k \rangle \langle h_l|\right) \\
&=\sum_{i=1}^{d_A}\sum_{j=1}^{d_A}\sum_{k=1}^{d_B}\sum_{l=1}^{d_B}a_{ij}^{kl}\mathcal{C}_A(|e_i \rangle\langle e_j|)\otimes \Tr\Psi_B(|h_k \rangle \langle h_l|) \\
&=\sum_{k=1}^{d_B}a_{ij}^{kk}\left(\sum_{i=1}^{d_A}\sum_{j=1}^{d_A}\mathcal{C}_A(|e_i \rangle\langle e_j|)\right) \\
&= \mathcal{C}_A\left(\sum_{k=1}^{d_B}\sum_{i=1}^{d_A}\sum_{j=1}^{d_A}a_{ij}^{kk}|e_i \rangle\langle e_j|\right) \\
&= \mathcal{C}_A(\Tr_B(\rho^{AB})).
\end{align*}
Therefore, we have 
\begin{align*}
\mathcal{Q}^{AB}(\Psi(\rho^{AB}))& =\mathcal{Q}( \mathcal{C}_A(\Tr_B(\rho^{AB}))) \\
&\leqslant \mathcal{Q}(\Tr_B(\rho^{AB})) \\
&=\mathcal{Q}^{AB}(\rho^{AB}),
\end{align*}

where the inequality  follows by ii) of the definition of QST measure.

Now consider arbitrary $\sum\limits_i p_i \rho^{AB}_i $ for any probability distribution $\{p_i\}$ ($\sum\limits_i p_i = 1$ and $p_i \geqslant 0$). We have
\begin{align*}
\mathcal{Q}^{AB}\left(\sum_i p_i \rho^{AB}_i\right) &=\mathcal{Q}\left(\sum_i p_i \Tr_B(\rho^{AB}_i)\right)\\
&\leqslant p_i \mathcal{Q}(\Tr_B(\rho^{AB}_i))\\
& = \sum_i p_i \mathcal{Q}^{AB}(\rho^{AB}_i),
\end{align*}
where the inequality  follows by iii) of the definition of QST measure.
Thus, the function $\mathcal{Q}^{AB}$ is a well-defined  QSBT measure on the system $\mathcal{H}_A\otimes \mathcal{H}_B$. 
\end{proof}

As an example, with respect to the state texture measure
$\mathfrak{R}(\rho)=-\ln\langle f_1 |\rho| f_1\rangle$ called state rugosity \cite{parisio 1},
we can define a state block texture measure,
\begin{equation}
\mathfrak{R}^{AB}(\rho^{AB})=\mathfrak{R}(\Tr_B(\rho^{AB}))=-\ln\langle f_1^A |\Tr_B(\rho^{AB})| f_1^A\rangle
\end{equation}
for arbitrary bipartite state $\rho^{AB}$. Theorem \ref{t4.1} shows that $\mathfrak{R}^{AB}$ is a well-defined QSBT measure.

Similarly, all the QST measures from Refs. \cite{wang2025, zhang2025,Muthuganesan2025} can be leveraged to define the corresponding QSBT measures. We list all these QST measures and their corresponding QSBT measures in Table \ref{tab_01}.

According to \cite[Theorem 3.]{wang2025}, the geometric measure $T^{AB}_g$ is an upper bound of the square of  the measure based on  trace distance $(T^{AB}_{\tr})^2$.

\begin{rem}
According to the definition of QSBT measure in Eq. (\ref{e7}), we have $\mathcal{Q}^{AB}(\rho^{AB})=\mathcal{Q}^{AB}(\sigma^{AB}) $ if  $\Tr_B(\rho^{AB})=\Tr_B(\sigma^{AB})$.
\end{rem}

\begin{eg} 
Let $\rho^1_p=p|\phi^+\rangle\langle \phi^+|+(1-p)|\psi^-\rangle\langle\psi^-|$
and $\rho^2_p=p|01\rangle\langle 01|+(1-p)|10\rangle\langle10|\)
with \(p\in (0,\frac{1}{2})\cup (\frac{1}{2},1)$ be two state families on $\C^2\otimes \C^2$, where $|\phi^+\rangle=\frac{1}{\sqrt{2}}(|00\rangle+|11\rangle)$ and $|\psi^-\rangle=\frac{1}{\sqrt{2}}(|01\rangle-|10\rangle)$.
Then $\Tr_B(\rho^1_p)=\frac{\mathbb{I}}{2}$ and $\Tr_B(\rho^2_p)=p|0\rangle\langle 0|+(1-p)|1\rangle\langle 1|$.

We have \(\mathfrak{R}^{AB}(\rho^1_p) = \mathfrak{R}^{AB}(\rho^2_p) =\ln2\). In other words, the block state rugosity cannot distinguish between these two state families. However,  by the trace distance, we have
\begin{equation*}
T^{AB}_{\tr}(\rho^1_p) =\frac{1}{2},\quad
T^{AB}_{\tr}(\rho^2_p) = \sqrt{p^2-p+\frac{1}{2}}.
\end{equation*}
When $p \in (0,\frac{1}{2})\cup (\frac{1}{2},1) $, $T^{AB}_{\tr}(\rho^1_p)<T^{AB}_{\tr}(\rho^2_p)\). Fig.\ref{Fig_01} shows that the trace-distance measure distinguishes all \(\rho^1_p\) and \(\rho^2_p\) for \(p \in (0,\frac{1}{2})\cup (\frac{1}{2},1) \). 

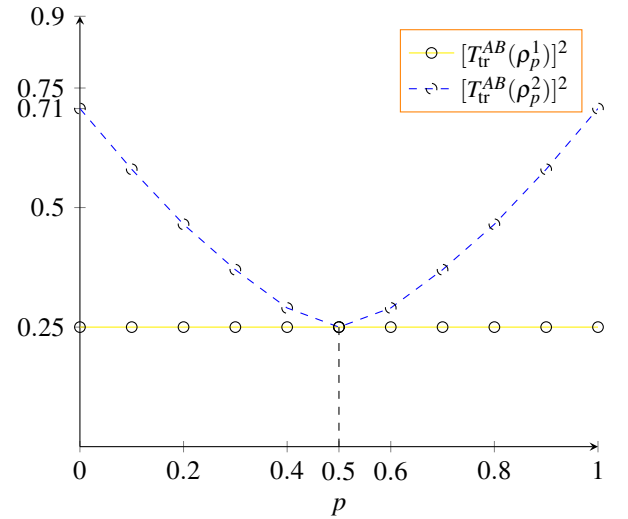
\begin{figure}[htbp]
\begin{tikzpicture}
\begin{axis}[
    axis lines=left,
    xlabel={$p$},
    ylabel={},
    xmin=0, xmax=1.0,
    ymin=0, ymax=0.9,
    xtick={0,0.2,0.4,0.5,0.6,0.8,1.0},
    ytick={0.25,0.5,0.707,0.75,0.9},
    grid=none,
    legend pos=north east,
    legend style={draw=orange, fill=white, font=\small},
]

\addplot[
    yellow,
    solid,
    mark=o,
    mark options={fill=white, draw=black},
    samples=20,
] coordinates {
    (0, 0.25)
    (0.1, 0.25)
    (0.2, 0.25)
    (0.3, 0.25)
    (0.4, 0.25)
    (0.5, 0.25)
    (0.6, 0.25)
    (0.7, 0.25)
    (0.8,0.25)
    (0.9, 0.25)
    (1.0, 0.25)
};
\addlegendentry{$[T^{AB}_{\tr}(\rho_p^1)]^2$}

\addplot[
    blue,
    dashed,
    mark=o,
    mark options={fill=white, draw=black},
    samples=20,
] coordinates {
    (0, 0.707)
    (0.1, 0.58)
    (0.2, 0.465)
    (0.3, 0.37)
    (0.4, 0.29)
    (0.5, 0.25)
    (0.6, 0.29)
    (0.7, 0.37)
    (0.8, 0.465)
    (0.9, 0.58)
    (1.0, 0.707)
};
\addlegendentry{$[T^{AB}_{\tr}(\rho_p^2)]^2$}

\addplot[
    black,
    only marks,
    mark=o,
    mark size=2pt,
] coordinates {(0.5, 0.25)};
\draw[dashed] (axis cs:0.5,0) -- (axis cs:0.5,0.25);
\end{axis}
 \end{tikzpicture}
\caption{$[T^{AB}_{\tr}(\rho_p^2)]^2$ and $[T^{AB}_{\tr}(\rho_p^1)]^2$ versus $p$.}
\label{Fig_01}
\end{figure}
\end{eg}

\subsection{Quantum-state block texture measures induced by a class of concave functions}
In this section, we present quantifiers of QSBT by the convex roof construction. Let 
$f:[0,1]\rightarrow[0,+\infty)$ be a function satisfying the following conditions:
\begin{enumerate}
    \item $f(1)=0$;
    \item $f$ is monotonically decreasing;
    \item $f$ is concave: $f(\lambda x+ (1-\lambda)y)\geqslant \lambda f(x)+(1-\lambda)f(y)$ for all $\lambda\in[0,1]$ and all $x,y\in[0,1]$.
\end{enumerate}

For any pure state $\ket{\psi^{AB}}$, we define
\begin{equation}\label{e09}
    \begin{split}
    \mathcal{P}_f(\ket{\psi^{AB}}) &= f\left(\Tr[\kb{f^A_1}{f^A_1}\Tr_B(\kb{\psi^{AB}}{\psi^{AB}})]\right) \\
    &= f\left(\bra{f^A_1}\Tr_B[\kb{\psi^{AB}}{\psi^{AB}}]\ket{f^A_1}\right),
    \end{split}
\end{equation}
and for any mixed state $\rho^{AB}$, we define
\begin{equation}
    \mathcal{P}_f(\rho^{AB}) = \min\sum_i p_i \mathcal{P}_f(\ket{\psi^{AB}_i}),
\end{equation}
where the minimization is taken over all ensembles $\{p_i, \ket{\psi^{AB}_i}\}$ such that $\rho^{AB}=\sum\limits_i p_i\kb{\psi^{AB}_i}{\psi^{AB}_i}$.

\begin{thm}\label{t3.4}
For any function $f$ satisfying the above conditions, the quantity $\mathcal{P}_f$ is a valid QSBT measure.
\end{thm}

The proof of this theorem is provided in the appendix A.

\begin{rem}
According to (\ref{e09}), for any pure states $\rho^{AB}$ and $\sigma^{AB}$ we have $\mathcal{P}_f(\rho^{AB})=\mathcal{P}_f(\sigma^{AB})$ if $\Tr_B(\rho^{AB})=\Tr_B(\sigma^{AB})$. 
\end{rem}

In fact, the QSBT measure $T^{AB}_g(\rho^{AB})=T_g(\Tr_B(\rho^{AB}))$ induced by the geometric measure of QST in the subsection A can be also constructed through a concave function $f(x)=1-x$ in the domain $[0,1]$ by Theorem III.4. Similarly, the QSBT measure $\mathfrak{R}^{AB}(\rho^{AB})=\mathfrak{R}(\Tr_B(\rho^{AB}))$ can be also obtained by using the concave function $f(x)=-\ln x$ in the domain $[0,1]$.

\begin{eg}
The function $f(x)=\cos(\frac{\pi}{2}x)$ is monotonically decreasing and concave on interval $[0,1]$ with $f(1)=\cos\frac{\pi}{2}$. Then we can define a QSBT measure $\mathcal{P}_{\cos}$ for any pure state $\ket{\psi^{AB}}$,
\begin{equation}\label{e011}
    \begin{split}
    \mathcal{P}_{\cos}(\ket{\psi^{AB}}) &= f\left(\Tr[\kb{f^A_1}{f^A_1}\Tr_B(\kb{\psi^{AB}}{\psi^{AB}})]\right) \\
    &= \cos\left(\bra{f^A_1}\Tr_B[\kb{\psi^{AB}}{\psi^{AB}}]\ket{f^A_1}\right),
    \end{split}
\end{equation}
and for any mixed state $\rho^{AB}$,
\begin{equation}
    \mathcal{P}_{\cos}(\rho^{AB}) = \min\sum_i p_i \mathcal{P}_{\cos}(\ket{\psi^{AB}_i}),
\end{equation}
where the minimization is taken over all ensembles $\{p_i, \ket{\psi^{AB}_i}\}$ such that $\rho^{AB}=\sum_i p_i\kb{\psi^{AB}_i}{\psi^{AB}_i}$.

Let us consider the following two-qubit pure states:
$$
\rho_1=\left(\begin{array}{llll}
\frac{1}{9}&\frac{2}{9}&0&\frac{2}{9}\\
\frac{2}{9}&\frac{4}{9}&0&\frac{4}{9}\\
0&0&0&0\\
\frac{2}{9}&\frac{4}{9}&0&\frac{4}{9}
\end{array}
\right),\quad
\rho_2=\left(\begin{array}{llll}
\frac{1}{9}&0&\frac{2}{9}&\frac{2}{9}\\
0&0&0&0\\
\frac{2}{9}&0&\frac{4}{9}&\frac{4}{9}\\
\frac{2}{9}&0&\frac{4}{9}&\frac{4}{9}
\end{array}
\right).
$$
We have
$$
\rho_1^A=\Tr_B(\rho_1)=\left(\begin{array}{lll}
\frac{5}{9}&\frac{4}{9}\\
\frac{4}{9}&\frac{4}{9}
\end{array}
\right),\quad\rho_2^A=\Tr_B(\rho_2)=\left(\begin{array}{lll}
\frac{1}{9}&\frac{2}{9}\\
\frac{2}{9}&\frac{8}{9}
\end{array}
\right).
$$
Thus $\mathcal{P}_{\cos}(\rho_1)=\cos(\frac{\pi}{2}\<f_1^A|\rho_1^A|f_1^A\>)=\cos(\frac{17\pi}{36})$ and $\mathcal{P}_{\cos}(\rho_2)=\cos(\frac{\pi}{2}\<f_1^A|\rho_2^A|f_1^A\>)=\cos(\frac{13\pi}{36}).$
Since $\mathcal{P}_{\cos}(\rho_1)<\mathcal{P}_{\cos}(\rho_2)$, then we can conclude that as two block states, 
 the texture of $\rho_2$ is larger than  that of $\rho_1$, and these two states can be distinguished in terms of block texture.
\end{eg}

\subsection{Quantum-state block measures from other approaches}

We first present a geometric measure of QSBT.
For a pure state $\ket{\psi^{AB}}$, we  define
    \begin{equation}\label{e8}
        M_g(\ket{\psi^{AB}}) = 1-\max_{\ket{\phi^{AB}}\in\mathcal{T}_B(\mathcal{H}_A\otimes\mathcal{H}_B)}\abs{\bk{\phi^{AB}}{\psi^{AB}}}^2.
    \end{equation}
     For a mixed state $\rho^{AB}$, we define $M_g$ as the minimal average block texture by convex roof extension,
    \begin{equation}\label{e9}
    M_g(\rho^{AB}) = \min\limits_{\{p_i,|\psi^{AB}_i\}}\sum_ip_iM_g(\ket{\psi^{AB}_i}),
    \end{equation}
    where the minimization is taken over all ensembles $\{p_i, \ket{\psi^{AB}_i}\}$ such that $\rho^{AB} = \sum\limits_ip_i\kb{\psi^{AB}_i}{\psi^{AB}_i}$.

\begin{prop}\label{p3.7}
$M_g$ is a well-defined QSBT measure.
\end{prop}

The proof is provided in the appendix B.

\begin{prop}\label{p3.8}
The following trace distance of QSBT is a bona fide measure,
    \begin{equation}\label{e15}
        M_{\mathrm{tr}}(\rho^{AB}) =\min_{\sigma^{AB}\in\mathcal{T}_B(\mathcal{H}_{A B})}\D(\rho^{AB},\sigma^{AB}),
    \end{equation}
    where $\D(\rho^{AB},\sigma^{AB})=\frac{1}{2}\norm{\rho^{AB}-\sigma^{AB}}_1$ and $\norm{X}_1 = \Tr\sqrt{XX^\dagger}$ is the trace norm of matrix $X$.
\end{prop}

The proof is provided in the appendix C.

To address this question whether $M_\text{tr}$ is the same as the trace-distance measure of QSBT $\mathcal{T}_\text{tr}^{AB}$ induced by the trace-distance measure of QST $\mathcal{T}_\text{tr}$, we present a concrete example below to show explicitly that these two measures are not the same.

\begin{eg}
Consider the quantum state
\[\rho^{AB}=\frac{1}{4}
\begin{pmatrix}
1 & 0 & 0 & \frac{1}{2}-\frac{1}{2}\mathrm{i} \\
0 & 1 & \frac{1}{2}-\frac{1}{2}\mathrm{i} & 0 \\
0 & \frac{1}{2}+\frac{1}{2}\mathrm{i} & 1 & 0 \\
\frac{1}{2}+\frac{1}{2}\mathrm{i} & 0 & 0 & 1
\end{pmatrix}.
\]

We have
\begin{align*}
\mathcal{T}_\mathrm{tr}^{AB}(\rho^{AB})&=\mathcal{T}_\mathrm{tr}(\Tr_B(\rho^{AB}))=\frac{1}{2}\norm{\Tr_B(\rho^{AB})-f^A}_1 \\
&=\frac{1}{2}\norm{\frac{\mathbb{I}}{2}-\frac{1}{2}\begin{pmatrix}1 & 1 \\1 & 1\end{pmatrix}}_1=\frac{1}{2}\norm{\begin{pmatrix}0 & -\frac{1}{2} \\-\frac{1}{2} & 0\end{pmatrix}}_1
=\frac{1}{2},
\end{align*}
and
\[\begin{array}{lllll}
M_\mathrm{tr}(\rho^{AB})&=\min_{\sigma^B\in\mathcal{D}(\mathbb{C}^2)}\mathsf{D}(\rho^{AB},f^A\otimes\sigma^B)\\
&=\min_{\sigma^B\in\mathcal{D}(\mathbb{C}^2)}\frac{1}{2}\norm{\rho^{AB}-f^A\otimes\sigma^B}_\mathrm{tr},
\end{array}
\]
where $f^A=\kb{f^A}{f^A}, \ket{f^A}=\frac{1}{\sqrt{2}}(\ket{0}+\ket{1})$. In the computational basis $\{\ket{0},\ket{1}\}$ of a single-qubit system, any single-qubit state $\sigma^B$ can be represented as:
\begin{equation*}
\sigma^B=\frac{1}{2}(\mathbb{I}_2+\textbf{r}\cdot\boldsymbol{\tau})=\frac{1}{2}\begin{pmatrix}1+z & x-\mathrm{i}y \\ x+\mathrm{i}y & 1-z\end{pmatrix},
\end{equation*}
where $\textbf{r}=(x,y,z)$ is a unit real vector and $\boldsymbol{\tau}=(\tau_x,\tau_y,\tau_z)$ is the vector given by the standard Pauli matrices $\tau_x,\tau_y,\tau_z$.

By direct calculation, we have

\begin{widetext}
$$\rho^{AB}-f^A\otimes\sigma^B=\frac{1}{4}
\begin{pmatrix}
-z & -x+\mathrm{i}y & -z-1 & -x+\frac{1}{2}+\mathrm{i}(y-\frac{1}{2}) \\
-x-\mathrm{i}y & z & -x+\frac{1}{2}-\mathrm{i}(y+\frac{1}{2}) & z-1 \\
-z-1 & -x+\frac{1}{2}+\mathrm{i}(y+\frac{1}{2}) & -z & -x+\mathrm{i}y \\
-x+\frac{1}{2}-\mathrm{i}(y-\frac{1}{2}) & z-1 & -x-\mathrm{i}y & z
\end{pmatrix}.
$$
\end{widetext}
The trace distance $\mathsf{D}(\rho^{AB},f^A\otimes\sigma^B)$ between $\rho^{AB}$ and $f^A\otimes\sigma^B$ is the sum of the absolute values of all the eigenvalues of the matrix $\rho^{AB}-f^A\otimes\sigma^B$. We calculate the trace distance for different randomly generated free states using MATLAB, see Fig. \ref{fig_02}. It is verified that $M_\mathrm{tr}(\rho^{AB})>T_\mathrm{tr}^{AB}(\rho^{AB})$.
\end{eg}
Generally, it is difficult  to estimate the geometric measure of  QSBT $M_g$. The following connection between the geometric measure $M_g$ and the trace distance $M_{tr}$ leads to a lower bound of $M_g$.
\begin{figure}[htbp] \label{fig_02}
  \centering
    \includegraphics[width=0.5\textwidth]{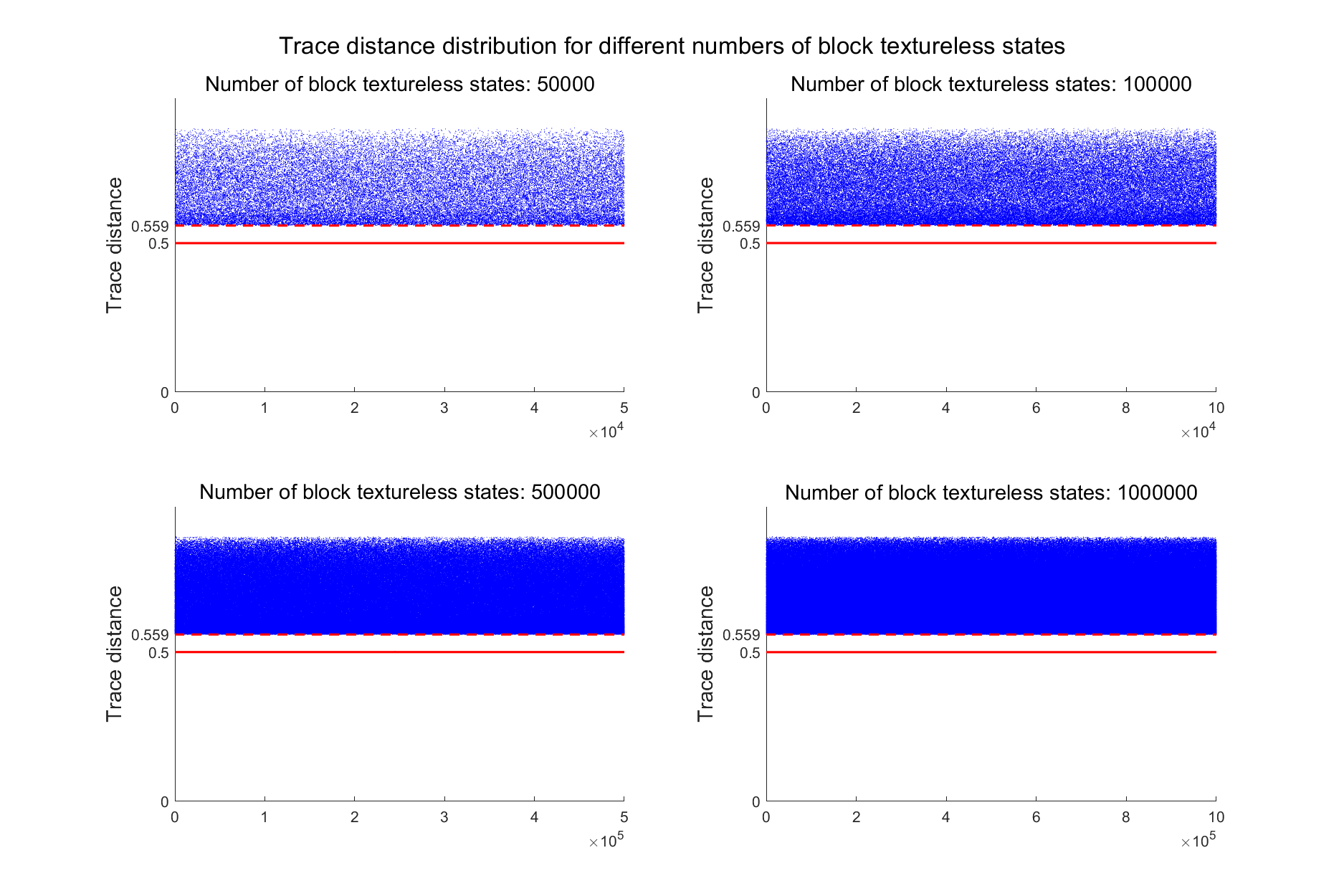}
    \caption{The red solid line represents the trace-distance measure of QSBT $\mathcal{T}_\mathrm{tr}^{AB}(\rho^{AB})$ induced by the trace-distance measure of QST $\mathcal{T}_\mathrm{tr}$. The blue points 
    in the four subplots represent the trace distance between the state $\rho^{AB}$ and the $5\times 10^4, 10^5, 5\times 10^5$ and $10^6 $ randomly generated block textureless states, respectively.}
    \label{fig_02}
\end{figure}
\begin{thm}\label{t5}
For any quantum state $\rho^{AB}$, we have 
\begin{equation}\label{e13}
M_g(\rho^{AB})\geqslant [M_{\tr}(\rho^{AB})]^2.
\end{equation}
\end{thm}

\begin{proof}
For any two pure states $|\psi\rangle_{AB}$ and $|\phi\rangle_{AB}$, we have \cite{Coles2019}
\begin{equation}\label{e17}
\D(|\psi\rangle_{AB}, |\phi\rangle_{AB})^2 = \frac{1}{2} \norm{\kb{\psi^{AB}}{\psi^{AB}} - \kb{\phi^{AB}}{\phi^{AB}}}_{F}, 
\end{equation}
where $\norm{\cdot}_F$ is the Frobenius norm. Let $M_g(\ket{\psi^{AB}})=1-\abs{\bk{\bar{\phi}^{AB}}{\psi^{AB}}}^2$, where $\ket{\bar{\phi}^{AB}}\in\mathcal{T}_B(\mathcal{H}_A\otimes\mathcal{H}_B)$. Note that for the pure state $\ket{\bar{\phi}^{AB}}\in \mathcal{T}_B(\mathcal{H}_A\otimes\mathcal{H}_B)$,
\begin{align*}
\norm{\ket{\psi^{AB}} - \ket{\bar{\phi}^{AB}}}_F&=\Tr\left[\left(\kb{\psi^{AB}}{\psi^{AB}} -\kb{\bar{\phi}^{AB}}{\bar{\phi}^{AB}}\right)^2\right] \\
&= 2-2\Tr\left(\kb{\psi^{AB}}{\psi^{AB}}\cdot\kb{\bar{\phi}^{AB}}{\bar{\phi}^{AB}}\right) \\
&= 2 - 2\left| \bk{\psi^{AB}}{\bar{\phi}^{AB}} \right|^2.
\end{align*}
Substituting the above equation into Eq. (\ref{e17}), we get
\begin{equation} \label{e18}
\D(|\psi\rangle_{AB}, |\bar{\phi}^{AB}\rangle)^2 = M_g(|\psi\rangle_{AB}).
\end{equation}
Thus for general state $\rho^{AB}$, we have
\begin{align*}
M_g(\rho^{AB}) 
&= \sum_i p_i \D\left(|\psi^{AB}_i\rangle, |\bar{\phi}^{AB}_i\rangle\right)^2 \\
&\geqslant\left[\sum_i p_i \D(|\psi^{AB}_i\rangle, |\bar{\phi}^{AB}_i\rangle)\right]^2 \\
&\geqslant \D\left(\rho^{AB}, \sum_i p_i \kb{\bar{\phi}^{AB}_i}{\bar{\phi}^{AB}_i}\right)^2
\geqslant [M_{\tr}(\rho^{AB})]^2,
\end{align*}
where $\rho^{AB} = \sum\limits_i p_i |\psi^{AB}_i\rangle\langle\psi^{AB}_i|$ is the optimal pure state decomposition corresponding to $M_g$, $\sum\limits_i p_i \kb{\bar{\phi}^{AB}_i}{\bar{\phi}^{AB}_i}\in\mathcal{T}_B(\mathcal{H}_A\otimes\mathcal{H}_B)$, the second equality is due to Eq. (\ref{e18}), and the last inequality is from the definition of $M_{\tr}$.
\end{proof}

Similarly, we can also define the QSBT measure based on fidelity,
\begin{equation}
    M_F(\rho^{AB}) = 1-\max_{\sigma^{AB}\in\mathcal{T}_B(\mathcal{H}_A\otimes\mathcal{H}_B)}\sqrt{F}(\rho^{AB},\sigma^{AB}),
\end{equation}
where $\sqrt{F}(\rho,\sigma) = \Tr\sqrt{\sqrt{\rho}\sigma\sqrt{\rho}}$ is the root fidelity between the states $\rho$ and $\sigma$ \cite{nielsen, khatri2020}.

With respect to the relative entropy of QSBT, we have the following conclusion, see proof in appendix D.

\begin{prop}\label{p3.11}
The relative entropy of QSBT defined by
\begin{equation}\label{e20}
M_r(\rho^{AB}) = \min_{\sigma^{AB}\in\mathcal{T}_B(\mathcal{H}_{A B})}S(\rho^{AB}\|\sigma^{AB})
\end{equation}
is a strongly monotonic QSBT measure, where $S(\rho\|\sigma) = \Tr(\rho\log\rho) - \Tr(\rho\log\sigma)$ is the quantum relative entropy between the states $\rho$ and $\sigma$.
\end{prop}

Next, we present an analytical formula for the relative entropy measure of QSBT.

\begin{thm}
The relative entropy of QSBT $M_r$ can be expressed as
\begin{equation}\label{e21}
M_r(\rho^{AB})=S(f^A_1\otimes\Tr_A\rho^{AB})-S(\rho^{AB}),
\end{equation}
where $S(\rho)= \Tr(\rho\log\rho)$ is the von Neumann entropy.
\end{thm}
\begin{proof}
From the definition in Eq. (\ref{e20}), we have
\begin{widetext}
\begin{align*}
M_r(\rho^{AB}) &= \min_{\tau^B\in\mathcal{D}(\mathcal{H}_B)}S(\rho^{AB}\|f^A_1\otimes\tau^B) \\
&=\min_{\tau^B\in\mathcal{D}(\mathcal{H}_B)}\{\Tr(\rho^{AB}\log\rho^{AB})-\Tr(\rho^{AB}\log(f^A_1\otimes\tau^B))\} \\
&= \min_{\tau^B\in\mathcal{D}(\mathcal{H}_B)}\{\Tr(\rho^{AB}\log\rho^{AB})-\Tr((f^A_1\otimes\Tr_A\rho^{AB})\log(f^A_1\otimes\tau^B))-\Tr((\rho^{AB}-f^A_1\otimes\Tr_A\rho^{AB})\log(f^A_1\otimes\tau^B))\} \\
&=\min_{\tau^B\in\mathcal{D}(\mathcal{H}_B)}\left\{\Tr(\rho^{AB}\log\rho^{AB})-\Tr((f^A_1\otimes\Tr_A\rho^{AB})\log(f^A_1\otimes\Tr_A\rho^{AB}))+\Tr((f^A_1\otimes\Tr_A\rho^{AB})\log(f^A_1\otimes\Tr_A\rho^{AB}))\right.\\
&\quad \left.-\Tr((f^A_1\otimes\Tr_A\rho^{AB})\log(f^A_1\otimes\tau^B))-\Tr((\rho^{AB}-f^A_1\otimes\Tr_A\rho^{AB})\log(f^A_1\otimes\tau^B))\right\}\\
&=\min_{\tau^B\in\mathcal{D}(\mathcal{H}_B)}\{-S(\rho^{AB})
+S(f^A_1\otimes\Tr_A\rho^{AB})+S(f^A_1\otimes\Tr_A\rho^{AB}\|f^A_1\otimes\tau^B)
-\Tr((\rho^{AB}-f^A_1\otimes\Tr_A\rho^{AB})\log(f^A_1\otimes\tau^B))\},
\end{align*}
in which the term
\begin{equation}\label{e22}
\Tr((\rho^{AB}-f^A_1\otimes\Tr_A\rho^{AB})\log(f^A_1\otimes\tau^B))=0
\end{equation}
for all $\tau^B\in\mathcal{D}(\mathcal{H}_B)$, see proof in the appendix E. Hence,
\begin{equation*}
M_r(\rho^{AB})=\min_{\tau^B\in\mathcal{D}(\mathcal{H}_B)}\{S(f^A_1\otimes\Tr_A\rho^{AB})
-S(\rho^{AB})+S(f^A_1\otimes\Tr_A\rho^{AB}\|f^A_1\otimes\tau^B)\}=S(f^A_1\otimes\Tr_A\rho^{AB})-S(\rho^{AB}).
\end{equation*}
\end{widetext}
\end{proof}
\begin{eg}
Let us consider the two-qubit Bell states, $\ket{\Phi^{\pm}}=(\ket{00}\pm\ket{11})/\sqrt{2},\ket{\Psi^{\pm}}=(\ket{01}\pm\ket{10})/\sqrt{2}$. By direct calculation, we have $\Tr_A(\ket{\Phi^{\pm}})=\Tr_A(\ket{\Psi^{\pm}})=\frac{\mathbb{I}_2}{2}$. Therefore, from Eq. (\ref{e21}) we obtain $M_r(\ket{\Phi^{\pm}})=M_r(\ket{\Psi^{\pm}})=1$. The block texture of the four Bell states are equal under the relative entropy measure. Note that the texture of the four Bell states is $+\infty$ \cite{wang2025}. In this sense the quantum-state block texture proposed in this paper gives rise to a more reasonable characterization than the one presented in Ref. \cite{wang2025}.
\end{eg}
\section{Conclusion}
We have generalized the quantum-state texture (QST) to quantum-state block texture (QSBT), and
established alternative frameworks for quantifying the block texture. For any given QST measure, we have constructed a corresponding QSBT measure based on the reduced systems. We have proposed measures for QSBT based on the concave functions and convex roof construction. We have also constructed several QSBT measures including the geometric measure of QSBT ($M_g$), the trace distance-based measure of QSBT ($M_{\tr}$), the fidelity-based measure of QSBT ($M_F$), and the relative entropy of QSBT ($M_r$). We have demonstrated that the geometric measure $M_g$ serves as an upper bound for the trace distance-based measure $M_{\tr}$. For a specific family of quantum states, we have evaluated the trace distance-based measures in detail. Our results show that these distinct QSBT measures play different roles in quantifying the block texture of quantum states.

In fact our results can be readily generalized to the case of arbitrary multipartite quantum states, by implementing bipartitions and thereby reducing the problem of characterizing multipartite quantum states to the bipartite scenario. Correspondingly, multiple definitions of block texture emerge for multipartite quantum states, which enables the systematic investigation of key issues including the monogamy and polygamy properties of quantum-state block texture (QSBT) measures. In general, QSBT furnishes a unifying and structure-driven framework for quantumnesses such as entanglement and other quantum correlations, with promising applications in classifying quantum phases, optimizing quantum simulation protocols and elucidating emergent phenomena in multipartite quantum systems.

\section*{Acknowledgments:}  
S.M. Fei acknowledges the financial support from specific research fund of the
Innovation Platform for Academicians of Hainan Province.

\section*{Data availability}
No data were created or analyzed in this study.

\section*{Appendices}
\subsection{Proof of Theorem \ref{t3.4}}
1. Non-negative: It is clear that for any pure state $\ket{\psi^{AB}}$, $\mathcal{P}_f(\ket{\psi^{AB}})\geqslant 0$, with the equality holding if and only if $\ket{\psi^{AB}}\in\mathcal{T}_B(\mathcal{H}_A\otimes\mathcal{H}_B)$.

2. Monotonicity: For any free operation $\Lambda\equiv\{K^{AB}_n\}$ and pure state $\ket{\psi^{AB}}$, one has \cite{parisio 1}
\begin{align}
    &\sum_n\Tr[\kb{f^A_1}{f^A_1}\Tr_B(K^{AB}_n\kb{\psi^{AB}}{\psi^{AB}}K^{AB\dagger}_n)] \nonumber \\
&\ \ \ \ \geqslant\Tr[\kb{f^A_1}{f^A_1}\Tr_B(\kb{\psi^{AB}}{\psi^{AB}})],
\end{align}
i.e.,
\begin{align}
    &\sum_n\bra{f^A_1}\Tr_B[K^{AB}_n\kb{\psi^{AB}}{\psi^{AB}}K^{AB\dagger}_n]\ket{f^A_1} \nonumber \\
&\ \ \ \ \geqslant\bra{f^A_1}\Tr_B[\kb{\psi^{AB}}{\psi^{AB}}]\ket{f^A_1}.
\end{align}
It follows that
\begin{widetext}
\begin{align*}
    \mathcal{P}_f(\psi^{AB})&= f(\bra{f^A_1}\Tr_B[\kb{\psi^{AB}}{\psi^{AB}}]\ket{f^A_1}) \\
    &\geqslant f\left(\sum_n\bra{f^A_1}\Tr_B[K^{AB}_n\kb{\psi^{AB}}{\psi^{AB}}K^{AB\dagger}_n]\ket{f^A_1}\right) \\
    &\geqslant \sum_n \Tr[K^{AB}_n\kb{\psi^{AB}}{\psi^{AB}}K^{AB\dagger}_n]f\left(\frac{\bra{f^A_1}\Tr_B[K^{AB}_n\kb{\psi^{AB}}{\psi^{AB}}K^{AB\dagger}_n]\ket{f^A_1}}{\Tr[K^{AB}_n\kb{\psi^{AB}}{\psi^{AB}}K^{AB\dagger}_n]}\right) \\
    &= \sum_n\Tr[K^{AB}_n\kb{\psi^{AB}}{\psi^{AB}}K^{AB\dagger}_n]\mathcal{P}\left(\frac{K^{AB}_n\kb{\psi^{AB}}{\psi^{AB}}K^{AB\dagger}_n}{\Tr[K^{AB}_n\kb{\psi^{AB}}{\psi^{AB}}K^{AB\dagger}_n]}\right).
\end{align*}
\end{widetext}

For a mixed state $\rho^{AB}$, let $\rho^{AB}=\sum\limits_i p_i\kb{\psi^{AB}}{\psi^{AB}}$ be an optimal pure-state ensemble, i.e.,
\begin{equation*}
    \mathcal{P}_f(\rho^{AB})=\sum_i p_i \mathcal{P}_f(\ket{\psi^{AB}}).
\end{equation*}
Thus, we have
\begin{widetext}
\begin{align*}
    \mathcal{P}_f(\Lambda(\rho^{AB})) &= \mathcal{P}(\sum_n K^{AB}_n\rho^{AB}K^{AB\dagger}_n) \\
    &=\mathcal{P}_f(\sum_n\sum_i p_i K^{AB}_n\kb{\psi^{AB}_i}{\psi^{AB}_i}K^{AB\dagger}_n) \\
    &=\mathcal{P}_f\left(\sum_n\sum_i p_i\Tr[K^{AB}_n\kb{\psi^{AB}_i}{\psi^{AB}_i}K^{AB\dagger}_n]\frac{K^{AB\dagger}_n\kb{\psi^{AB}_i}{\psi^{AB}_i}K^{AB\dagger}_n}{\Tr[K^{AB\dagger}_n\kb{\psi^{AB}_i}{\psi^{AB}_i}K^{AB\dagger}_n]}\right) \\
    &\leqslant \sum_n\sum_i p_i\Tr[K^{AB}_n\kb{\psi^{AB}_i}{\psi^{AB}_i}K^{AB\dagger}_n]\mathcal{P}\left(\frac{K^{AB}_n\kb{\psi^{AB}_i}{\psi^{AB}_i}K^{AB\dagger}_n}{\Tr[K^{AB}_n\kb{\psi^{AB}_i}{\psi^{AB}_i}K^{AB\dagger}_n]}\right) \\
    &\leqslant \sum_i p_i \mathcal{P}_f(\ket{\psi^{AB}_i}) =\mathcal{P}_f(\rho^{AB}),
\end{align*}
\end{widetext}
which implies the monotonicity of $\mathcal{P}$ under free operations.

3. Convexity: For any set of states $\{\rho^{AB}_i\}$ and any probability distribution $\{p_i\}$, let $\rho^{AB}_i=\sum_j q_{ij} \rho^{AB}_{ij}$ be the optimal pure-state ensemble for $\rho^{AB}_i$, i.e.,
$\mathcal{P}_f(\rho^{AB}_i)=\sum_j q_{ij}\mathcal{P}_f(\rho^{AB}_{ij}).$
Then
$$\begin{array}{llll}
    \mathcal{P}_f(\sum_i p_i\rho^{AB}_i)&= \mathcal{P}_f(\sum_i\sum_j p_i q_{ij}\rho^{AB}_{ij}) \\
    &\leqslant \sum_i \sum_j p_i q_{ij}\mathcal{P}_f(\rho^{AB}_{ij}) \\
    &=\sum_i p_i \mathcal{P}(\rho^{AB}_i).
\end{array}
$$

\subsection{Proof of Proposition \ref{p3.7}}
It is easy to see that for a pure state $\ket{\psi^{AB}}$, the definition of the geometric measure of QSPT can be rewritten as
\begin{equation}
    M_g(\ket{\psi^{AB}}) = 1-\max_{\ket{\gamma^{AB}}\in\mathcal{T}_B(\mathcal{H}_A\otimes\mathcal{H}_B)}F\left(\ket{\psi^{AB}},\ket{\gamma^{AB}}\right),
\end{equation}
where $F(\rho,\sigma) = \left(\Tr\sqrt{\sqrt{\rho}\sigma\sqrt{\rho}}\right)^2$ is the fidelity between the states $\rho$ and $\sigma$ \cite{nielsen,khatri2020}.
It is clear that $M_g(\ket{\psi^{AB}}) \geqslant 0$, with equality holding if and only if $\ket{\psi^{AB}} \in \mathcal{F}$. Let $M_g(\ket{\psi^{AB}})=1-F(\ket{\psi^{AB}},\ket{\overline{\gamma}^{AB}})$. For any free operation $\mathcal{E}$, we have
\begin{align*}
    M_g(\mathcal{E}(\ket{\psi^{AB}})) &= 1-\max_{\ket{\gamma^{AB}}\in\mathcal{T}_B(\mathcal{H}_A\otimes\mathcal{H}_B)}F(\mathcal{E}(\ket{\psi^{AB}}),\ket{\gamma^{AB}}) \\
    &\leqslant 1-F(\mathcal{E}(\ket{\psi^{AB}}),\mathcal{E}(\ket{\overline{\gamma}^{AB}})) \\
    &\leqslant 1-F(\ket{\psi^{AB}},\ket{\overline{\gamma}^{AB}})\\
    & = M_g(\ket{\psi^{AB}}).
\end{align*}

Let $\rho^{AB}=\sum_i p_i\kb{\overline{\psi}^{AB}_i}{\overline{\psi}^{AB}_i}$ be the optimal decomposition for a mixed state $\rho^{AB}$. Then $\mathcal{E}(\rho^{AB})=\sum_i p_i\mathcal{E}(\kb{\overline{\psi}^{AB}_i}{\overline{\psi}^{AB}_i})=\sum_i \sum_k p_i q_{i,k}\kb{\phi^{AB}_{i,k}}{\phi^{AB}_{i,k}}$, where $\mathcal{E}(\kb{\overline{\psi}^{AB}_i}{\overline{\psi}^{AB}_i})=\sum_k q_{i,k}\kb{\phi^{AB}_{i,k}}{\phi^{AB}_{i,k}}$ is the optimal decomposition of $\mathcal{E}(\kb{\overline{\psi}^{AB}_i}{\overline{\psi}^{AB}_i})$ for any $i$. Therefore,
\begin{align*}
    M_g(\mathcal{E}(\rho^{AB})) &\leqslant \sum_i p_i q_{i,k}M_g(\ket{\phi^{AB}_{i,k}}) \\
    &= \sum_i p_i M_g(\mathcal{E}(\kb{\overline{\psi}^{AB}_i}{\overline{\psi}^{AB}_i})) \\
    &\leqslant \sum_i p_i M_g(\ket{\overline{\psi}^{AB}_i}) \\
    &= M_g(\rho^{AB}).
\end{align*}

Finally, for any states $\{\rho_k^{AB}\}$ and any probability distribution $\{q_k\}$, let $\rho^{AB}_k=\sum_i p_{k,i}\kb{\overline{\psi}^{AB}_{k,i}}{\overline{\psi}^{AB}_{k,i}}$ be the optimal decomposition of $\rho^{AB}_k$ for all $k$. Thus, $\sum_k q_k \rho^{AB}_k =\sum_k\sum_i q_k p_{k,i} \kb{\overline{\psi}^{AB}_{k,i}}{\overline{\psi}^{AB}_{k,i}}$ is a decomposition of the state $\sum_k q_k \rho^{AB}_k$, and we have
\begin{align*}
    M_g\left(\sum_k q_k \rho^{AB}_k\right) &\leqslant \sum_k\sum_i q_k p_{k,i} M_g(\ket{\overline{\psi}^{AB}_{k,i}}) \\
    &= \sum_k q_k \sum_i p_{k,i} M_g(\ket{\overline{\psi}^{AB}_{k,i}}) \\
    &= \sum_k q_k M_g(\rho^{AB}_k).
\end{align*}

\subsection{Proof of Proposition \ref{p3.8}}
For any quantum state $\rho^{AB}$, it is evident that $M_{\mathrm{tr}}(\rho^{AB}) \geqslant 0$, with the equality holding if and only if $\rho^{AB} \in \mathcal{T}_B(\mathcal{H}_A\otimes\mathcal{H}_B)$. For any free operation $\Lambda$, we have
    \begin{align*}
        M_{\mathrm{tr}}(\Lambda(\rho^{AB}))
        &= \min_{\sigma^{AB}\in\mathcal{T}_B(\mathcal{H}_A\otimes\mathcal{H}_B)}\frac{1}{2}\norm{\Lambda(\rho^{AB})-\sigma^{AB}}_1 \\
    &\leqslant \min_{\sigma^{AB}\in\mathcal{T}_B(\mathcal{H}_A\otimes\mathcal{H}_B)}\frac{1}{2}\norm{\Lambda(\rho^{AB})-\Lambda(\sigma^{AB})}_1 \\
        &\leqslant \min_{\sigma^{AB}\in\mathcal{T}_B(\mathcal{H}_A\otimes\mathcal{H}_B)}\frac{1}{2}\norm{\rho^{AB}-\sigma^{AB}}_1 \\
        &= M_{\mathrm{tr}}(\rho^{AB}),
    \end{align*}
    where the first and the second equalities follow from the definition of $M_{\mathrm{tr}}$, the first inequality follows from the fact that $\Lambda(\sigma^{AB}) \in\mathcal{T}_B(\mathcal{H}_A\otimes\mathcal{H}_B)$, and the second inequality follows from the contractivity of the trace distance under completely positive and trace-preserving maps, i.e., $\norm{\Lambda(\rho)-\Lambda(\sigma)}_1 \leqslant \norm{\rho-\sigma}_1$ for any quantum states $\rho$ and $\sigma$.
    
Let $\sigma^*_i$ be the optomal block textureless state that minimizes the distance to $\rho_i^{AB}$ for all $i$. For any probability distribution $\{p_i\}$, we have
    \begin{align*}
        M_{\mathrm{tr}}\left(\sum_ip_i\rho_i^{AB}\right)
    &\leqslant \frac{1}{2}\norm{\sum_ip_i\rho_i^{AB}-\sum_ip_i\sigma^*_i}_1 \\
        &\leqslant \sum_ip_i\cdot\frac{1}{2}\norm{\rho_i^{AB}-\sigma^*_i}_1 \\
        &= \sum_ip_iM_{\mathrm{tr}}(\rho_i^{AB}),
    \end{align*}
    where the first inequality follows from the definition of $M_{\mathrm{tr}}$ and the fact that $\sum_i p_i \sigma^*_i$ is a textureless state, the second inequality follows from the convexity of the trace norm, and the equality follows from the definition of $\sigma^*_i$.       

\subsection{Proof of Proposition \ref{p3.11}}

Clearly, the relative entropy of QSBT is nonnegative and $M_r(\rho^{AB})=0$ if and only if $\rho^{AB} \in \mathcal{T}_B(\mathcal{H}_A\otimes\mathcal{H}_B)$ is a block textureless state. 

Let $\Lambda$ be a free operation. We have
\begin{align*}
M_r(\Lambda(\rho^{AB})) & =\min_{\sigma^{AB}\in\mathcal{T}_B(\mathcal{H}_A\otimes\mathcal{H}_B)}S(\Lambda(\rho^{AB})\|\sigma^{AB}) \\
&\leqslant \min_{\sigma^{AB}\in\mathcal{T}_B(\mathcal{H}_A\otimes\mathcal{H}_B)}S(\Lambda(\rho^{AB})\|\Lambda(\sigma^{AB})) \\
&\leqslant \min_{\sigma^{AB}\in\mathcal{T}_B(\mathcal{H}_A\otimes\mathcal{H}_B)} S(\rho^{AB}\|\sigma^{AB}) \\
&= M_r(\rho^{AB}),
\end{align*}
where the first and the second equalities follow from the definition of $M_r$, the first inequality follows from the fact that $\Lambda(\mathcal{T}_B(\mathcal{H}_A\otimes\mathcal{H}_B))\subseteq\mathcal{T}_B(\mathcal{H}_A\otimes\mathcal{H}_B)$, and the second inequality follows from the fact that the quantum relative entropy is contractive under CPTP maps \cite{lindblad1975,ruskai2002}.

For any quantum state $\rho^{AB}_i$, there exists a block textureless state $\sigma^{AB}_i\in\mathcal{T}_B(\mathcal{H}_A \otimes \mathcal{H}_B)$ such that
\begin{equation}
M_r(\rho^{AB}_i)=S(\rho^{\rho}_i\|\sigma^{AB}_i).
\end{equation}

Hence, we have
\begin{align*}
M_r(\sum_i p_i \rho^{AB}_i) &\leqslant S\left(\sum_i p_i \rho^{AB}_i \middle\| \sum_i p_i \sigma^{AB}_i\right) \\
&\leqslant \sum_i p_i S\left(\rho^{AB}_i \middle\| \sigma^{AB}_i\right) \\
&= \sum_i p_i M_r(\rho^{AB}_i),
\end{align*}
where the first inequality follows from the definition of $M_r$ and the fact that $\sum_i p_i \sigma^{AB}_i$ is a block textureless state, the second inequality follows from the fact that the quantum relative entropy is jointly convex \cite{lindblad1975,ruskai2002}.

As a quantum relative entropy satisfies the conditions (F1)-(F5) and Eq. (22) in Ref. \cite{vedral1998}, we have 
\begin{equation}
        S(\rho \|\sigma)\geqslant\sum_n p_n S\left(\rho_n\middle\|\frac{K_n \sigma K_n^\dagger}{\Tr(K_n \sigma K_n^\dagger)}\right),
\end{equation}
where $\rho_n= K_n\rho K_n^\dagger / p_n$ and $p_n = \Tr(K_n \rho K_n^\dagger)$ with $\{K_n\}$ is a set of Kraus operators corresponding to a CPTP map. Then we have
\begin{align*}
    M_r(\rho^{AB}) &= S(\rho^{AB}\|\sigma_*^{AB}) \\
    &\geqslant \sum_n p_n S\left(\rho^{AB}_n\middle\|\frac{K_n \sigma^{AB}_* K_n^\dagger}{\Tr(K_n \sigma^{AB}_* K_n^\dagger)}\right) \\
    &\geqslant \sum_n p_n \min_{\sigma^{AB}\in\mathcal{T}_B(\mathcal{H}_A\otimes\mathcal{H}_B)}S(\rho^{AB}_n\|\sigma^{AB}) \\
    &= \sum_n p_n M_r(\rho^{AB}_n).
\end{align*}

\subsection{Proof of Eq. \eqref{e22}}
We express $\rho^{AB}$ as
\begin{equation*}
\rho^{AB}=\sum_i\sum_j\ketbra{i}{j}_A\otimes\rho^B_{i,j},
\end{equation*}
so $\Tr_A\rho^{AB}=\sum_k\rho^B_{k,k}$.
Therefore, we have
\begin{equation*}
f^A_1\otimes\Tr_A\rho^{AB}=\sum_{i,j}\left(\ketbra{i}{j}_A\otimes\frac{1}{d_A}\sum_k\rho^B_{k,k}\right),
\end{equation*}
thus
\begin{equation*}
\rho^{AB}-f^A_1\otimes\Tr_A\rho^{AB}=\sum_{i,j}\left(\ketbra{i}{j}_A\otimes\left(\rho^B_{i,j}-\frac{1}{d_A}\sum_k\rho^B_{k,k}\right)\right).
\end{equation*}

On the support of the block textureless state $\sigma^{AB}=f^A_1\otimes\tau^B$, we can derive
\begin{align*}
\log(f^A_1\otimes\tau^B)&=(\log f^A_1)\otimes\mathbb{I}_B+\mathbb{I}_A\otimes(\log\tau^B) \\
&=\mathbb{I}_A\otimes(\log\tau^B) \\
&=\sum_m\left(\ketbra{m}{m}_A\otimes\log\tau^B\right),
\end{align*}
hence
\begin{align*}
&\quad(\rho^{AB}-f^A_1\otimes\Tr_A\rho^{AB})\log(f^A_1\otimes\tau^B) \\&=\sum_{i,j,k}\left(|i \rangle\langle j|m \rangle\langle m|_A\otimes\left(\rho^B_{i,j}-\frac{1}{d_A}\sum_k\rho^B_{k,k}\right)\log\tau^B\right).
\end{align*}
Taking the trace yields
\begin{align*}
&\quad\Tr((\rho^{AB}-f^A_1\otimes\Tr_A\rho^{AB})\log(f^A_1\otimes\tau^B)) \\
&=\Tr\left(\left[\sum_m\left(\rho^B_{m,m}-\frac{1}{d_A}\sum_k\rho^B_{k,k}\right)\right]\log\tau^B\right) \\
&=0
\end{align*}

\end{document}